 \documentclass[final,3p,times,twocolumn,authoryear]{elsarticle}

\usepackage{framed,multirow}

\usepackage{amssymb}
\usepackage{latexsym}
\usepackage{graphicx}
\usepackage{amsmath}
\usepackage{subfigure}
\usepackage{booktabs}
\usepackage{tabularx}
\usepackage{lineno}
\usepackage{dsfont}
\usepackage[table,xcdraw]{xcolor}
\usepackage[marginal]{footmisc}
\DeclareMathOperator*{\argmin}{argmin}

\DeclareUnicodeCharacter{0301}{\'{e}}
\newtheorem{theorem}{Theorem}
\newtheorem{remark}{Remark}
\newtheorem{proof}{Proof}
\usepackage{url}
\usepackage{xcolor}

\usepackage{hyperref}

\definecolor{newcolor}{rgb}{.8,.349,.1}

\begin{document}


\begin{frontmatter}

\title{Learning Homeomorphic Image Registration\\ via  Conformal-Invariant Hyperelastic Regularisation}%

\author[1,2]{Jing Zou \fnref{fn1}}
\fntext[fni]{Work done during a visiting in Cambridge}
\author[3]{Noémie Debroux}
\author[2]{Lihao Liu}
\author[1]{Jing Qin \corref{cor1}}
\cortext[cor1]{Corresponding author: harry.qin@polyu.edu.hk}
\author[2]{Carola-Bibiane Schönlieb}
\author[2]{Angelica I Aviles-Rivero}

\address[1]{Center for Smart Health, School of Nursing, The Hong Kong Polytechnic University, HKSAR, China}
\address[2]{Department of Applied Mathematics and Theoretical Physics, 
University of Cambridge, Cambridge, UK}
\address[3]{Université Clermont Auvergne, France}


\begin{abstract}
Deformable image registration is a fundamental task in medical image analysis and plays a crucial role in a wide range of clinical applications. 
Recently, deep learning-based approaches have been widely studied for deformable medical image registration and achieved promising results. However, existing deep learning image registration techniques do not  theoretically guarantee 
topology-preserving transformations. This is a key property to preserve anatomical structures and achieve plausible transformations that can be used in real clinical settings.
We propose a novel framework for deformable image registration. Firstly, we introduce a  novel regulariser based on conformal-invariant properties in a nonlinear elasticity setting.
Our regulariser enforces  the deformation field  to be  smooth, invertible and orientation-preserving. 
More importantly, we strictly guarantee topology preservation yielding to a clinical meaningful registration.  Secondly, 
we boost the performance of our regulariser through coordinate MLPs, where one can view the to-be-registered images as continuously differentiable entities. 
We demonstrate, through numerical and visual experiments, that our framework is able to outperform current techniques for image registration.
\end{abstract}


\end{frontmatter}


\section{Introduction}
Deformable image registration (DIR) is a crucial tool in modern healthcare, allowing physicians to compare and combine information from multiple images of the same patient or different patients. It is a fundamental task in medical images analysis and plays a key role in various clinical applications including image-guided interventions~\citep{tam2016image}, radiotherapy~\citep{krilavicius2016predicting, brock2017use}
diagnosis~\citep{tekchandani2022computer}, and treatment planning~\citep{krilavicius2016predicting}.

DIR aims to align two or more images by optimising a non-linear voxel-wise deformation mapping between the given target and source images.
Traditional methods treat DIR as a pair-wise optimisation problem, which relies on specific modelling assumptions. 
As a result, flexibility is limited, and it is difficult to handle complex and unpredictable deformations. Furthermore, the selection of an appropriate model, that accurately represents the deformation, can be challenging.
In addition, such iterative optimisation process is often time-consuming, which limits their applicability in real-time clinical scenarios.

Learning-based approaches have proliferated and have been widely studied for DIR and achieved promising results e.g.~\citep{cciccek20163d,heinrich2019closing,long2015fully,ronneberger2015u,liu2019probabilistic,zou2022deformable}.
Following this research line, several techniques and strategies have been proposed to improve registration performance. For example, cascaded networks~\citep{zhao2019recursive}, pyramid coarse-to-fine networks~\citep{mok2020large} or joint models~\citep{liu2021rethinking}.
However, the improvements are limited
due to two main challenges. Firstly, existing models assume to have, for the training stage, a large and well-represented dataset. However, in medical domain, acquiring such dataset is challenging. Additionally, even with a large dataset, the representations learned from training data cannot capture the optimal complex deformations for all image pairs. Secondly, 
current learning-based methods cannot theoretically and strictly guarantee diffeomorphism of the deformation field.
Although $L^2$ regularisation \citep{balakrishnan2019voxelmorph,hering2021cnn}, cycle consistency constraint \citep{kim2021cyclemorph} and inverse consistency constraint \citep{zhu2022swin,liu2022pc} are utilised to improve the quality of deformation fields, 
these frameworks strictly cannot theoretically guarantee diffeomorphic transformations--this is reflected in reporting negative values in the Jacobian determinant.

To handle the aforementioned challenges, optimising pair-wise transformations has been explored from the variational and learning perspectives. This set of techniques mitigates the need of having a large dataset. 
A crucial factor, for any type of image registration methods, is the regulariser used to enforce plausible transformations.  For example, the community largely has used Laplacian constraints, incompressibility constraints, and fluid-like regularisers~\citep{passieux2012high,christensen1996deformable,mansi2011ilogdemons}. However, these regularisers fail to handle large deformations or do not capture the characteristics of the target tissues. Perhaps hyperelastic regularisation~\citep{veress2005measurement,phatak2009strain,burger2013hyperelastic} is the most widely used constraint in the medical domain; as this principle matches with the characteristics of biological tissues. Learning-based methods for pair-wise registration have also used such types of regularisers. For example, the work of that~\citep{wolterink2022implicit} used the existing hyperelastic regulariser of~\citep{burger2013hyperelastic}. A weak constraint has been imposed on the Jacobian determinant in~\citep{han2023diffeomorphic}. In contrast to those works that use existing regularisers. In this work, we focus on proposing a novel regulariser that is more flexible and can provide theoretical guarantees. 

We summarise the contributions as follows:
\begin{itemize}

    \item We propose a novel pair-wise image registration framework that eliminates the need for pre-training or prior affine registration.

    \item The key contribution of this paper is that we introduce a novel conformal-invariant regulariser. The proposed regulariser unifies the hyperelasticity with conformal and Beltrami-based approaches. It enables simultaneous control of changes in length, area, and volume and of distortion with conformal mappings, as well as smoothness of deformations. Most importantly, \textit{our regulariser yields to homeomorphic transformations enforcing plausible transformations and theoretical guarantees.}

    \item We introduce a learned image deformation mapping that is driven by our novel regulariser via coordinate MLP, and the learning process is resolution-independent, offering more flexibility. 

    \item We demonstrate, through extensive experimental results, that our proposed framework yields to better performance than existing implicit and explicit regularisers.

\end{itemize}

\section{Related Works}
In this section, we thoroughly examine the existing literature and highlight the distinctive aspects of our work in comparison.

\subsection{Traditional Methods}
Over the past few decades, the field of image registration has primarily focused on traditional methods, leading to the development of toolboxes such as Elastix~\citep{klein2009elastix} and ANTs \citep{avants2009advanced}. These methods typically involve defining a transformation model, an optimisation function, and a similarity metric, followed by an iterative parameter optimisation process to find the optimal transformation. Various mathematical models have been utilized in these methods, including elastic-type models \citep{bajcsy1989multiresolution}, free-form deformation models based on cubic B-splines \citep{rueckert1999nonrigid}, statistical parametric mapping \citep{ashburner2000voxel}, and Demons-based approaches \citep{pennec1999understanding}.
In order to ensure smoother and invertible deformation fields, diffeomorphic transformations have gained significant attention in the past decades \citep{oliveira2014medical}. 
Prominent examples of diffeomorphic transformation methods include Large deformation diffeomorphic metric mapping (LDDMM)\citep{beg2005computing,cao2005large,glaunes2008large}, Symmetric Normalization (SyN)\citep{rogelj2006symmetric,beg2007symmetric,avants2008symmetric} (integrated into ANTs), and diffeomorphic Demons \citep{vercauteren2009diffeomorphic}. Despite their achievements, these methods are computationally intensive due to their reliance on iterative optimisation strategies. In addition, selecting an appropriate transformation model that accurately captures the deformations between input image pairs is still challenging.

\subsection{Learning-based Methods}
In learning-based methods, the deformation mapping from the source image to the target image is typically parameterised via deep networks, and the mapping is learned from lots of training image pairs. The main advantage of learning-based methods is the remarkable speed in the inference phase once the training is complete.

Supervised learning methods usually rely on ground truth deformation fields, which are either obtained from traditional methods \citep{yang2017fast,fan2019birnet} or synthesized through manual deformable transformation \citep{sokooti2017nonrigid,sokooti20193d}.
However, there are challenges associated with obtaining ground truth deformation fields. On one hand, generating real deformation fields can be time-consuming and computationally expensive. On the other hand, using pseudo ground-truth deformation fields that obtained through manual deformable transformations may not accurately represent the true distribution of misalignments in real-world scenarios.

To handle the challenges, researchers have turned to unsupervised methods, whose training relies on image-wise similarities rather than the ground truth deformation fields. 
The popular VoxelMorph~\citep{balakrishnan2019voxelmorph} is based on Unet for deformable brain MR image registration. 
With only paired images as input, VoxelMorph efficiently predicts the corresponding deformation field, and the network is updated by minimizing the intensity-based discrepancy between the warped and the target images. Besides, other strategies are proposed to improve the performance of learning-based registration methods. For example, cascaded networks \citep{zhao2019recursive,zhao2019unsupervised}, multi-stage pyramid network \citep{de2019deep,mok2020large}, and weak supervisions \citep{xu2019deepatlas,hering2021cnn}. In cascaded networks, the source image is progressively warped by the deformation field from each cascade network, finally aligned to the target image. The final deformation field is the composition of all fields from these cascaded networks. In multi-stage pyramid network, the deformation field is learned in a coarse-to-fine manner, where the registration is performed at multiple scales of image pyramids. For weak supervisions, the segmentation masks or the anatomical landmarks are utilized to guide the learning.

Although these methods have demonstrated improvement of the registration performance, there still exist challenges for learning-based methods.  
Firstly, the training of these networks usually requires a large dataset. However, acquiring a well-represented and sufficiently large dataset in the medical domain is challenging. 
Secondly, they can not theoretically guarantee the diffeomorphism of the learned deformation field.

\subsection{Regularisations in Registration}
In deformable medical image registration, anatomical structures often exhibit local and non-rigid variations. 
Achieving accurate registration necessitates a deformation field capable of capturing these variations at a high resolution with a high degree of freedom.
The most challenging part of such registration is to ensure that the estimated deformation field between images is smooth and realistic, while avoiding overfitting due to the high degree of freedom. 
The key to a successful registration is the choice of the regularisation as it determines the types of transformations that the model can generate.

In traditional registration methods, many regularisations are proposed and they are commonly model-dependent. In the elastic model, the first-order derivatives or the second-order derivatives of the deformation field are penalised to ensure the smoothness \citep{modersitzki2003numerical}.
Regularisation for the B-spline model is performed by applying constraints on the displacements of the B-spline control points \citep{rueckert1999nonrigid}. For fluid models, the regularisation is employed by  penalising spatial derivatives of the velocity fields \citep{vialard2012diffeomorphic}.
Additionally, laplacian constraints, incompressibility constraints, and fluid-like regularisers~\citep{passieux2012high,christensen1996deformable,mansi2011ilogdemons} are proposed for specific registration tasks.


Many learning-based registration methods employ a simple and generic regularization, which takes the form of a penalty on the first-order derivative of the deformation field~\citep{balakrishnan2018unsupervised,balakrishnan2019voxelmorph,hering2021cnn,zhao2019recursive}. The performance is limited as it assumes uniform smoothness properties irrespective of image content and allowable deformations.
Furthermore, some networks utilised 
cycle consistency constraint \citep{kim2019unsupervised,kim2021cyclemorph} and inverse consistency constraint \citep{zhu2022swin,liu2022pc} to improve the quality of the deformation field. However, these constraints are performed in a bi-directional manner, which introduces additional computational complexity. Moreover, they can not theoretically guarantee plausible transformations.

\section{Methodology}

\begin{figure*}[!ht]
\centering
\setlength{\abovecaptionskip}{0pt}
\setlength{\belowcaptionskip}{0pt}
\includegraphics[width=\linewidth]{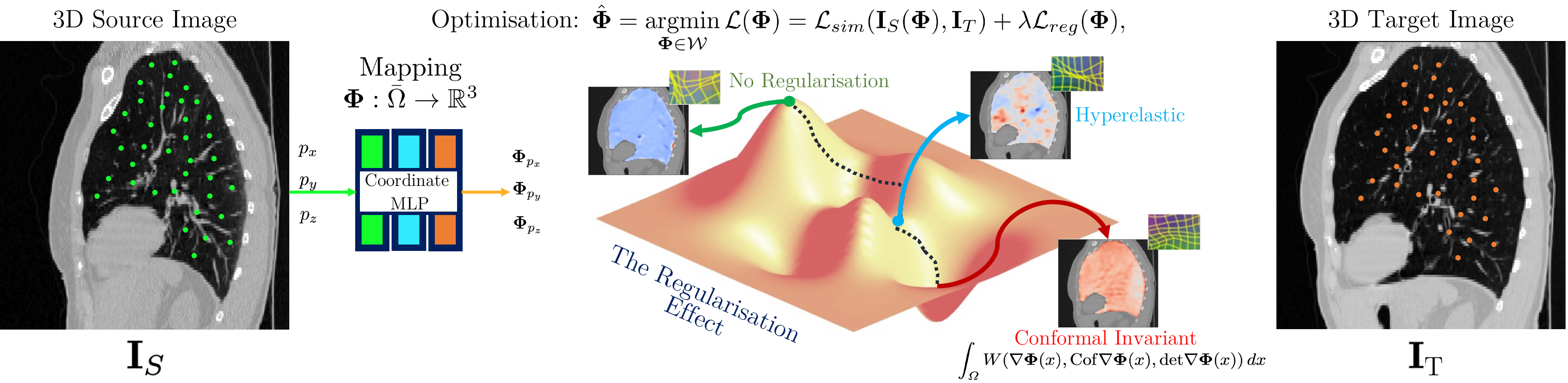}
\caption{Our proposed workflow framework. We seek to optimise, through a coordinate MLP, the mapping $\mathbf{\Phi}$ to align the coordinates between the source and target images. Our highlight is a new regulariser, whose effect is illustrated in the middle part. Our proposed conformal-invariant hyperelastic regulariser enforces volume presentation, controls changes in length and area, and ensures smoothness of deformation yielding to a better optimisation outcome.}
\label{FigOverview}
\vspace{0pt}
\end{figure*}
This section describes our novel image registration framework. It contains two key parts: 
i) the introduction and motivation of our new conformal-invariant hyperelastic-based regulariser to impose a physically relevant nature on the deformation, and ii) the details of our learning framework based on coordinate Multi-Layer Perceptron (MLP). 
The overall workflow framework is displayed in Fig. \ref{FigOverview}.

\subsection{Problem Statement}

In this work, we consider the problem of how to improve 3D image registration.  Let $\Omega$ be a convex bounded open subset of $\mathbb{R}^3$ (the image spatial domain) of class $\mathcal{C}^1$ thus satisfying the cone property meaning that there exists a finite cone $C$ such that each point $x\in \Omega$ is the vertex of a finite cone $C_x$ contained in $\Omega$ and congruent to $C$, and theoretically required to ensure Ball's results \citep{ball}. Let $\mathbf{I}_S:\Omega \to \mathbb{R}$ be the source image to be deformed, and $\mathbf{I_T}:\Omega \to \mathbb{R}$ be the target image to be fixed. The mapping $\mathbf{\Phi}: \bar{\Omega}\to\mathbb{R}^3$ denotes the sought non-rigid and non-parametric deformation aligning the deformed source image $\mathbf{I}_S(\mathbf{\Phi})$ with the target one $\mathbf{I}_T$, while $\nabla \mathbf{\Phi}: \bar{\Omega}\to M_3(\mathbb{R})$ is the Jacobian deformation with $M_3(\mathbb{R})$ the set of $3\times3$ matrices.

Since image registration is highly ill-posed, in an unsupervised deep learning framework in which we place ourselves, the design of the cost function is crucial. Following arguments from the variational setting, the sought deformation is obtained by minimising a loss function comprising two key terms. Firstly,
a similarity term that measures how close the deformed source image is to the target image in a sense depending on the application. Secondly, a regularisation term encoding the prescribed nature of the allowed deformations: 

\begin{equation}
    \hat{\mathbf{\Phi}} = \argmin_{\mathbf{\Phi} \in \mathcal{W}}\mathcal{L}(\mathbf{\Phi})=\mathcal{L}_{sim}(\mathbf{I}_S(\mathbf{\Phi}),\mathbf{I}_T)+\lambda\mathcal{L}_{reg}(\mathbf{\Phi}), 
    \label{loss_function}
\end{equation}

\noindent
$\lambda$ is a trade-off parameter between both terms. 

\subsection{A Novel Conformal-Invariant Regularisation}
Inspired by the theory of mechanics, the shapes to be matched are viewed as physical bodies subjected to forces and undergoing deformations. In this context, a deformation \citep{ciarlet1994three} is a smooth mapping that is orientation preserving and injective except possibly on the boundary $\partial \Omega$ if self-contact is allowed which mathematically translates into $\mathrm{det}\nabla \mathbf{\Phi}>0$ almost everywhere. 
From the variational perspective, it is well-understood that in order to obtain meaningful solutions, one needs to enforce the problem to be well-posed. The key to such is the regulariser.

The majority of learning-based image registration models leverage the empirical properties of a given architecture to approximate the well-posedness of the registration. Whilst they report great performance, they do not provide any guaranteed property yielding to topologically preserving transformations. 
Motivated by this drawback in the learning-based literature, we introduce a novel regulariser that guarantees homeomorphic transformations-- this is translated to have a model that enforces topology preservation yielding to clinically meaningful output. With the previous motivation in mind, our new model falls within the hyperelasticity setting allowing large and non-linear deformations \citep[ch. 4]{ciarlet1994three} while keeping an elastic behavior and computational performance. 

Our motivation to select this principle is that organs and several biological phenomena are characterised by hyperelasticity. We note that highly deformable materials such as rubber, filled elastomer, or biological tissues, are often modelled within this hyperelasticity setting. Unlike linear elasticity assuming small strains and the validity of Hooke's law (linear relation between stress and strains), the hyperelasticity theory predicates the existence of a stored energy density function differentiable with respect to the deformation in each direction and whose derivative gives the state of stress within the material in the same direction. The total deformation stored energy is given by the integral of this density over the whole body and by minimising this quantity \textit{we find a solution to the equilibrium problem} \citep{ledret2004}. 

In our model, the objects are considered as isotropic (the material deforms in the same way in every direction), homogeneous (meaning that the behavior is the same everywhere inside the material), and hyperelastic materials allowing large compressions and expansions while keeping a mechanical elastic behavior. This interpretation drives the construction of our new regularisation term therefore based on the stored energy of such materials prescribing then a physically-meaningful nature for the generated deformations. Our new regulariser reads:
\begin{eqnarray}
\mathcal{L}_{reg}(\mathbf{\Phi}) &=& \int_\Omega W(\nabla \mathbf{\Phi}(x),\mathrm{Cof}\nabla \mathbf{\Phi}(x),\mathrm{det}\nabla \mathbf{\Phi}(x))\,dx\\  \nonumber
&+& \mathds{1}_{\{\|.\|_{L^\infty(\Omega)}\leq \beta\}}(\mathrm{det}\nabla \Phi),  \label{eq:reg}
\end{eqnarray}

with $\beta>0$, and $W(\nabla \mathbf{\Phi},\mathrm{Cof}\nabla \mathbf{\Phi},\mathrm{det}\nabla \mathbf{\Phi}) = \left\{ \begin{array}{l} \frac{a_1\|\nabla \mathbf{\Phi}\|_F^9}{(\mathrm{det}\nabla \mathbf{\Phi})^3}+\frac{a_2\|\mathrm{Cof}\nabla \mathbf{\Phi}\|_F^6}{(\mathrm{det}\nabla \mathrm{\Phi})^4} + a_3(\mathrm{det}\nabla \mathbf{\Phi}-1)^2 + \text{...}\\ \text{...} \frac{a_4}{(\mathrm{det}\nabla \mathbf{\Phi})^\alpha} - 3^{\frac{9}{2}}a_1 - 3^3a_2-a_4 \text{ if } \mathrm{det}\nabla \mathbf{\Phi} >0 \\ +\infty  \text{ otherwise}\end{array} \right.$, \\

\noindent 
$\|A\|_F = \sqrt{\mathrm{Tr}(A^TA)} $ being the Frobenius norm, Cof denotes the cofactor matrix, and $\alpha>1$, $a_1>0$, $a_2>0$, $a_3>0$ and $a_4>0$ being hyper-parameters. {The first two terms are distorsion mappings and measure the deviation from conformality of the deformation and are conformally invariant meaning that any conformal change of variables does not alter their values \citep{adamowicz2007}}. Our regulariser is then seen as the unification of the hyperelasticity with the conformal and Beltrami-based approaches \citep{lam2014}. 

They also control the changes in length and area as well as the smoothness of the deformations. The third and fourth terms are added to control changes in volume that is promoting volume-preserving deformations by penalising deviations from one of the Jacobian determinant while preventing singularities by penalising small values of it. The last constants are added to impose $W(I_3,I_3,1)=0$ as required by the hyperelasticity setting. The stored energy density function $W$ is thus polyconvex and rotation-invariant by definition, and the deformations are searched in the following suitable functional space $\mathcal{W} = \{\varphi \in \mathrm{Id}+W^{1,9}_0(\Omega,\mathbb{R}^3)\, |  \, \frac{\|\nabla \varphi\|_F^9}{(\mathrm{det}\nabla \varphi)^3}\in L^1(\Omega), \frac{\|\mathrm{Cof}\nabla \varphi\|_F^6}{(\mathrm{det}\nabla \varphi)^4}\in L^1(\Omega),\mathrm{det}\nabla \varphi \in L^2(\Omega),\frac{1}{\mathrm{det}\nabla \varphi}\in L^\alpha(\Omega), \nabla \det \varphi >0 \text{ a.e. in } \Omega,\,\|\mathrm{det}\nabla \varphi\|_{L^\infty(\Omega)}\leq \beta \}$, that is the deformation is assumed to be equal to the identity on the boundary $\partial \Omega$.
\begin{remark}
This assumption is lifted in  practice since the computation of the deformation is done only in the lung mask inside the image spatial domain $\Omega$ as explained in the following section.
\end{remark}
\begin{remark}
    The $L^\infty$ penalisation on the Jacobian determinant is added only for theoretical purposes. In practice, $\beta$ is taken as the maximum value allowed in computers and the penalisation is lifted in the numerical algorithm.
\end{remark}
Since $9>4$ and 
\begin{align}
 &\int_\Omega \|(\nabla \mathbf{\Phi})^{-1}(x)\|^4\mathrm{det}\nabla \mathbf{\Phi}(x)\,dx \\  \label{eq:hom} \nonumber
&= \int_\Omega \|\frac{1}{\mathrm{det}\nabla \mathbf{\Phi}(x)}\mathrm{Cof}\nabla \mathbf{\Phi}(x)^T \|^4 \mathrm{det}\nabla \mathbf{\Phi}(x)\,dx,\\ \nonumber
&=\int_\Omega \frac{1}{(\mathrm{det}\nabla \mathbf{\Phi}(x))^3}\|\mathrm{Cof}\nabla \mathbf{\Phi}(x)\|^4\,dx,\\ \nonumber
&=\int_\Omega \frac{\|\mathrm{Cof}\nabla \mathbf{\Phi}(x)\|^4}{(\mathrm{det}\nabla \mathbf{\Phi})^{\frac{8}{3}}}\frac{1}{(\mathrm{det}\nabla \mathbf{\Phi}(x))^{\frac{1}{3}}}\,dx,\\ \nonumber
&\leq \Bigg(\int_\Omega\frac{\|\mathrm{Cof}\nabla \mathbf{\Phi}(x)\|^6}{(\mathrm{det}\nabla \mathbf{\Phi}(x))^4}\,dx\Bigg)^{\frac{2}{3}}\Bigg(\int_\Omega\frac{1}{\mathrm{det}\nabla \mathbf{\Phi}(x)}\,dx\Bigg)^{\frac{1}{3}} \\ \nonumber
&\text{ by Holder's inequality with $p=\frac{3}{2}$ and $q=3$},
\end{align}
then Ball's results apply \citep{ball} and the deformations are homeomorphisms with the inverse deformation $\mathbf{\Phi}^{-1}\in W^{1,4}(\Omega,\mathbb{R}^3)$.

Regarding $\mathcal{L}_{sim}$ from~\eqref{loss_function}, in this work, we take the similarity term as the classical normalised crossed-correlation between the deformed source image and the target one:
\begin{align*}
    &\mathcal{L}_{sim}(\mathbf{I}_S(\mathbf{\Phi}),\mathbf{I}_T) = -NCC(\mathbf{I}_S(\mathbf{\Phi}),\mathbf{I}_T) \\
    &=-\frac{\langle \mathbf{I}_S\circ\Phi,\mathbf{I}_R\rangle_{L^2(\Omega)}}{\|\mathbf{I}_S\circ\Phi\|_{L^2(\Omega)}\|\mathbf{I}_R\|_{L^2(\Omega)}}\\
    &= -\sum_{p\in \Omega}\frac{(\sum_{p_i\in w}(\mathbf{I}_S(\mathbf{\Phi}(p_i))-\bar{\mathbf{I}}_S(\mathbf{\Phi}(p)))(\mathbf{I}_T(p_i)-\bar{\mathbf{I}}_T(p)))^2}{\sum_{p_i\in w}(\mathbf{I}_S(\mathbf{\Phi}(p_i))-\bar{\mathbf{I}}_S(\mathbf{\Phi}(p)))^2\sum_{p_i\in w}(\mathbf{I}_T(p_i)-\bar{\mathbf{I}}_T(p))^2},
\end{align*}
with $w$ being a local window of size $n^3$ around the current pixel $p$, and $\bar{\mathbf{I}}_S(\mathbf{\Phi})(p)$ and $\bar{\mathbf{I}}_T(p)$ are the mean intensity of that local window.

\subsection{Theoretical results}
\begin{theorem}[Existence of minimisers]
Problem \eqref{loss_function} admits at least one minimiser in $\mathcal{W}$.
\end{theorem}
\begin{proof}
    The proof follows the arguments of classical direct methods of the calculus of variations. We first derive a coercivity inequality. Using the fact that $(a-b)^2 \geq \frac{1}{2}a^2-b^2$, one has : 
    \begin{align*}
        \mathcal{L}(\Phi) &\geq \frac{a_1}{\beta^3}\|\nabla \Phi\|_{L^9(\Omega,M_3(\mathbb{R}))}^9 + \frac{a_2}{\beta^4}\|\mathrm{Cof}\nabla \Phi\|_{L^6(\Omega,M_3(\mathbb{R}))}^6\\
        &+ \frac{a_3}{2}\|\mathrm{det}\nabla \Phi\|_{L^2(\Omega)}^2-a_3\mathrm{meas}(\Omega)+a_4\|\frac{1}{\mathrm{det}\nabla \Phi}\|_{L^\alpha(\Omega)}^\alpha\\
        &-3^{\frac{9}{2}}a_1\mathrm{meas}(\Omega)-3^2a_2\mathrm{meas}(\Omega)-a_4\mathrm{meas}(\Omega) \\
        &+ \mathds{1}_{\{\|.\|_{L^\infty(\Omega)}\leq \beta\}}(\mathrm{det}\nabla \Phi),
    \end{align*}
    with $\mathrm{meas}$ standing for the measure of the space.
    The quantity $\mathcal{L}(\Phi)$ is thus bounded below and as for $\Phi=\mathrm{Id}$, the identity map, and suitable $\beta$, $\mathcal{L}(\Phi)=-NCC(\mathbf{I}_S,\mathbf{I}_T)$ is finite with $\mathbf{I}_S$ and $\mathbf{I}_T$ continuous. The infimum is thus finite.\\
    Let then $(\Phi_k)_k \in \mathcal{W}$ be a minimising sequence, i.e. $\underset{k\rightarrow +\infty}{\lim}\mathcal{L}(\Phi_k)=\underset{\Phi\in\mathcal{W}}{\inf} \mathcal{L}(\Phi)$. Hence there exists $K\in \mathbb{N}$ such that $\forall k \in \mathbb{N}$, $k\geq K$ implies $\mathcal{L}(\Phi_k)\leq \underset{\Phi\in\mathcal{W}}{\inf}\mathcal{L}(\Phi)+1$. From now on, we assume that $k \geq K$. According to the coercivity inequality, one gets :
    \begin{itemize}
        \item $(\Phi_k)$ is uniformly bounded according to $k$ in  $W^{1,9}(\Omega,\mathbb{R}^3)$ using the generalised Poincar\'e inequality and the fact that $\Phi=\mathrm{Id}$ on $\partial \Omega$. 
        \item $(\mathrm{Cof}\nabla \Phi_k)$ is uniformly bounded according to $k$ in $L^6(\Omega,M_3(\mathbb{R}))$.
        \item $(\mathrm{det}\nabla \Phi_k)$ is uniformly bounded according to $k$ in $L^\infty(\Omega)$ and in $L^2(\Omega)$.
        \item $(\frac{1}{\mathrm{det}\nabla \Phi_k})$ is uniformly bounded according to $k$ in $L^\alpha(\Omega)$.
    \end{itemize}
    Thus there exist a subsequence -- still denoted by $(\Phi_k)$ -- and $\bar{\Phi}\in W^{1,9}(\Omega,\mathbb{R}^3)$ such that $\Phi_k \underset{k\rightarrow +\infty}{\rightharpoonup} \bar{\Phi}$ in $W^{1,9}(\Omega,\mathbb{R}^3)$. Moreover, there exists a subsequence (common with the previous one which is always possible) -- still denoted by $(\mathrm{Cof}\nabla \Phi_k)$ --  and $\gamma\in L^6(\Omega,M_3(\mathbb{R}))$ such that $\mathrm{Cof}\nabla \Phi_k \underset{k\rightarrow +\infty}{\rightharpoonup} \gamma $ in $L^6(\Omega,M_3(\mathbb{R}))$. Furthermore, there exists a subsequence also common with previous ones which is always possible -- still denoted by $(\mathrm{det}\nabla \Phi_k)$-- and $\delta \in L^2(\Omega)$ such that $\mathrm{det}\nabla \Phi_k \underset{k\rightarrow+\infty}{\rightharpoonup}\delta$ in $L^2(\Omega)$.
    
    By applying Theorem 8.20 from \citep{dacorogna}, 
    then $\mathrm{Cof}\nabla \Phi_k\underset{k\rightarrow +\infty}{\rightharpoonup}\mathrm{Cof}\nabla \bar{\Phi}$ in $L^{\frac{9}{2}}(\Omega,M_3(\mathbb{R}))$, and by uniqueness of the weak limit as well as the continuous embedding of $L^p$-spaces, we have that $\gamma = \mathrm{Cof}\nabla \bar{\Phi}$ in $L^6(\Omega,M_3(\mathbb{R}))$ and thus $\mathrm{Cof}\nabla \Phi_k\underset{k\rightarrow +\infty}{\rightharpoonup}\mathrm{Cof}\nabla \bar{\Phi}$ in $L^6(\Omega,M_3(\mathbb{R}))$. We also get that $\mathrm{det}\nabla \Phi_k \underset{k\rightarrow +\infty}{\rightharpoonup} \mathrm{det}\nabla \bar{\Phi}$ in $L^3(\Omega)$ and by uniqueness of the weak limit as well as the continuous embedding of $L^p$-spaces, we get that $\delta = \mathrm{det}\nabla \bar{\Phi}$ in $L^2(\Omega) $ and therefore $\mathrm{det}\nabla \Phi_k\underset{k\rightarrow +\infty}{\rightharpoonup}\mathrm{det}\nabla \bar{\Phi}$ in $L^2(\Omega)$ and $\mathrm{det}\nabla \Phi_k \underset{k\rightarrow+\infty}{\overset{*}{\rightharpoonup}}\mathrm{det}\bar{\Phi}$ in $L^\infty(\Omega)$. At last, by continuity of the trace operator, we have that $\bar{\Phi}\in \mathrm{Id}+W^{1,9}_0(\Omega,\mathbb{R}^3)$. The stored energy function $W$ is continuous and convex: if $\psi_n\underset{n\rightarrow +\infty}{\longrightarrow}\bar{\psi}$ in $W^{1,9}(\Omega,\mathbb{R}^3)$, thus $\nabla \psi_n \underset{n\rightarrow+\infty}{\longrightarrow}\nabla \bar{\psi}$ in $L^9(\Omega,M_3(\mathbb{R}))$ and one can extract a subsequence still denoted by $(\nabla \psi_n)$ such that $\nabla \psi_n\underset{n\rightarrow+\infty}
    {\longrightarrow}\nabla \bar{\psi}$ almost everywhere in $\Omega$. 
    
    Similarly, if $K_n\underset{n\rightarrow+\infty}{\longrightarrow}\bar{K}$ in $L^6(\Omega,M_3(\mathbb{R}))$ then one can extract a common subsequence still denoted by $(K_n)$ such that $K_n\underset{n\rightarrow+\infty}{\longrightarrow}\bar{K}$ almost everywhere in $\Omega$. Finally, if $\delta_n\underset{n\rightarrow+\infty}{\longrightarrow} \bar{\delta}$ in $L^2(\Omega)$ then one can extract 
    a common subsequence still denoted by $(\delta_n)$ such that $\delta_n\underset{n\rightarrow+\infty}{\longrightarrow} \bar{\delta}$ almost everywhere in $\Omega$. Then, by continuity of $W$, one gets that $W_(\nabla \psi_n,K_n,\delta_n)\underset{n\rightarrow+\infty}{\longrightarrow}W(\nabla \bar{\psi},\bar{K},\bar{\delta})$ almost everywhere in $\Omega$. Then applying Fatou's lemma yields $\int_\Omega W(\nabla \bar{\psi},\bar{K},\bar{\delta})\,dx\leq \underset{k\rightarrow +\infty}{\lim \inf} \int_{\Omega}W(\nabla \psi_n,K_n,\delta_n)\,dx $. 
    As $W$ is convex, so is $\int_\Omega W(\xi,K,\delta)\,dx$, and involving Br\'ezis Corollaire III.8 \citep{brezis}
    leads to: $\int_\Omega W(\nabla \bar{\Phi},\mathrm{Cof}\nabla\bar{\Phi},\mathrm{det}\nabla \bar{\Phi})\,dx\leq \underset{k\rightarrow+\infty}{\lim\inf} \int_\Omega W(\nabla \Phi_k,\mathrm{Cof}\nabla \Phi_k,\mathrm{det}\nabla \Phi_k)\,dx<+\infty $.
    Since $W(\nabla \bar\Phi,\mathrm{Cof}\nabla \bar{\Phi},\mathrm{det}\nabla \bar{\Phi}) = +\infty$ when $\mathrm{det}\nabla \bar{\Phi}(x)\leq 0$ , the set on which it occurs is necessarily of null measure otherwise we would have $\mathcal{L}(\bar{\Phi})=+\infty$. So $\mathrm{det}\nabla \bar{\Phi})>0$ almost everywhere on $\Omega$. Besides, for all $k\geq K$, $\int_\Omega\|(\nabla \Phi_k)^{-1}\|_F^4\mathrm{det}\nabla \Phi_k\,dx\leq C$.
    $C=C(\alpha,\beta,\Omega)>0$ being a constant depending only on $\alpha, \beta$ and $\Omega$. The assumptions of Ball's results thus hold \citep{ball} and it yields that $(\Phi_k$ are homeomorphisms from $\bar{\Omega}$ to $\bar{\Omega}$ and $\Phi_k^{-1}\in W^{1,4}(\Omega,\mathbb{R}^3)$.
    By the weak-$*$ lower semi-continuity of $\|.\|_{L^\infty(\Omega)}$, we deduce that $\|\mathrm{det}\nabla \bar{\Phi}\|_{L^\infty(\Omega)}\leq \underset{k\rightarrow+\infty}{\lim\inf}\|\mathrm{det}\nabla \Phi_k\|_{L^\infty(\Omega)}\leq \beta$ and $\mathds{1}_{\{\|.\|_{L^\infty(\Omega)}\leq \beta\}}(\nabla \bar{\Phi})\leq \underset{k\rightarrow+\infty}{\lim \inf}\mathds{1}_{\{\|.\|_{L^\infty(\Omega)}\leq \beta\}}(\nabla \Phi_k)  $. Also, $\int_\Omega \|(\nabla \bar{\Phi})^{-1}\|^4\mathrm{det}\nabla \bar \Phi\,dx <\infty$ following the same calculus as previously. Therefore $\bar{\Phi}$ is also a homeomorphim from $\bar{\Omega}$ to $\bar{\Omega}$. 
    Since $\mathbf{I}_S\in L^2(\Omega)\cap \mathcal{C}^0(\Omega)$, we have that $\mathbf{I}_S\circ\Phi_k \underset{k\rightarrow+\infty}{\longrightarrow}\mathbf{I}_S\circ\bar{\Phi}$ almost everywhere in $\Omega$ and thus $NCC(\mathbf{I}_S\circ\Phi_k,\mathbf{I}_R)\underset{k\rightarrow+\infty}{\longrightarrow}NCC(\mathbf{I}_S\circ\bar{\Phi},\mathbf{I}_R)$ by using Lebesgue theorem and assuming both $\mathbf{I}_S$ and $\mathbf{I}_R$ are bounded. We thus have proved that $\mathcal{L}(\bar{\Phi})\leq \underset{k\rightarrow+\infty}{\lim\inf}\mathcal{L}(\Phi_k)=\underset{\Phi\in\mathcal{W}}{\inf}\mathcal{L}(\Phi)<+\infty$ with $\bar{\Phi}\in \mathcal{W}$.
\end{proof}
\subsection{Learning Transformations via Coordinate MLP} 

In our pair-wise registration framework, we employ a coordinate MLP to continuously represent the transformation. Unlike traditional methods that directly operate on image intensities, we use the spatial coordinates $(x, y, z)$ within the lung area as inputs to the MLP. 
In our lung CT registration task, due to sophisticated structures in the lungs, we expect that the coordinate MLP can represent high-frequency information, which means local small deformations.
To achieve this, we adopt the strategy inspired by SIREN \citep{sitzmann2020implicit}, which utilizes periodic sinusoidal functions to model high-frequency content effectively. 
The SIREN strategy is based on the observation that neural networks with sinusoidal activation functions can approximate high-frequency signals more efficiently than traditional activation functions such as ReLU or Sigmoid. By using sinusoidal activations, the network can better represent fine-grained details and capture high-frequency variations.

Take a 3D position $p\in \mathbb{R}^3$ in $\mathbf{I}_T$ as input, set $p'\in \mathbb{R}^3 $ as corresponding position in the coordinate space of $\mathbf{I}_S$.
The deformation between $p$ and $p'$ is parameterised as following:
\begin{equation}
    \mathbf{\Phi}(p)=sin(\omega(W_ip+b_i)),
\end{equation}
where $\omega$ is the hyperparameter that regulates
the spectral bias of the network, we set it as 32 for all our experiments.
In our setting, we use a 4-layer MLP for the COPD dataset and a 3-layer MLP for 4DCT dataset with 256 hidden units.


\section{Experiments and Results}
This section describes in details the range of experiments conducted to evaluate our proposed framework.

\subsection{Experimental Setup}
In our experimental setup, we introduce the datasets utilised and present our comprehensive implementation details.

\subsubsection{Datasets Description} We evaluate our framework on  inspiration-to-expiration lung CT registration task using two publicly available dataset: DIRLab COPD~\citep{castillo2013reference} and DIRLab 4DCT~\citep{castillo2009framework}.
Both the two datasets are composed of lung CT images characterising the inspiration-to-expiration motion of the lung from 10 patients. They also
provide 300 anatomical landmarks (serve as ground truth for the evaluation) for images that at the inspiration phase and at the expiration phase.
More details please refer to the provided website link\footnote{The dataset is released at the website \url{https://med.emory.edu/departments/radiation-oncology/research-laboratories/deformable-image-registration/index.html}}.

\begin{table*}[!h]\small
\renewcommand\arraystretch{1}
\centering
\caption{Numerical comparison of our proposed framework vs. other existing image registration techniques. The numerical values reflect the TRE metric $(mm)$ for the DIRLab COPD dataset. The best results are highlighted in green colour.}
\begin{center}
\begin{tabular}{p{0.45in} p{0.55in}p{0.55in}
                          p{0.55in}p{0.55in}
                          p{0.645in}p{0.55in}
                          p{0.5in}}
\toprule
\rowcolor[HTML]{EFEFEF}
\ COPD &{\ \ \ $Init.$}&{\ FE}&{\ PDD}&{\ \  VM }
  &{LapIRN} &{ \ \ INR}&{\ \ \ Ours}  \\
\midrule
{\ \ \ \ 01} &{\ \ \ 26.33}  &{\ \ \ 4.89} &{\ \ \ 2.57} &{\ \ \ 9.95} &{\ \ \ \ 6.85} &{\ \ \ 2.53} &{\ \ \ \cellcolor[HTML]{D7FFD7}\textbf{1.66}} \\

{\ \ \ \ 02} &{\ \ \ 21.79}  &{\ \ \ 7.30} &{\ \ \ 4.01} &{\ \ \ 9.96} &{\ \ \ \ 6.90} &{\ \ \ 5.78} &{\ \ \ \cellcolor[HTML]{D7FFD7}\textbf{3.70}}\\

{\ \ \ \ 03} &{\ \ \ 12.64}  &{\ \ \ 2.89} &{\ \ \ 1.46} &{\ \ \ 4.41} &{\ \ \ \ 1.51} &\cellcolor[HTML]{D7FFD7}{\ \ \ \textbf{1.28}} &{\cellcolor[HTML]{D7FFD7}\ \ \ \textbf{1.28}}\\

{\ \ \ \ 04} &{\ \ \ 29.58}  &{\ \ \ 5.46} &{\ \ \ 2.19} &{\ \ \ 7.08} &{\ \ \ \ 6.38} &{\ \ \ 2.34} &\cellcolor[HTML]{D7FFD7}{\ \ \ \textbf{1.62}}\\

{\ \ \ \ 05} &{\ \ \ 30.08}  &{\ \ \ 5.19} &{\ \ \ 2.22} &{\ \ \ 9.19} &{\ \ \ \ 6.81} &{\ \ \ 3.09} &{\ \ \ \cellcolor[HTML]{D7FFD7}\textbf{1.47}}\\

{\ \ \ \ 06} &{\ \ \ 28.46}  &{\ \ \ 5.53} &{\ \ \ 1.89} &{\ \ \ 8.12} &{\ \ \ \ 4.19} &{\ \ \ 2.66} &{\ \ \ \cellcolor[HTML]{D7FFD7}\textbf{1.86}}\\

{\ \ \ \ 07} &{\ \ \ 21.60}  &{\ \ \ 4.40} &{\ \ \ 1.62} &{\ \ \ 7.10} &{\ \ \ \ 2.73} &{\ \ \ 1.27} &{\ \ \ \cellcolor[HTML]{D7FFD7}\textbf{1.20}}\\

{\ \ \ \ 08} &{\ \ \ 26.46}  &{\ \ \ 3.94} &{\ \ \ 1.72} &{\ \ \ 7.92} &{\ \ \ \ 4.32} &{\ \ \ 2.75} &{\ \ \ \cellcolor[HTML]{D7FFD7}\textbf{1.65}}\\

{\ \ \ \ 09} &{\ \ \ 14.86}  &{\ \ \ 3.57} &{\ \ \ 1.51} &{\ \ \ 6.93} &{\ \ \ \ 3.60} &{\ \ \ 1.40} &{\ \ \ \cellcolor[HTML]{D7FFD7}\textbf{1.30}}\\

{\ \ \ \ 10} &{\ \ \ 21.81}  &{\ \ \ 4.44} &{\ \ \ 2.43} &{\ \ \ 9.16} &{\ \ \ \ 6.59} &{\ \ \ 3.25} &{\ \ \ \cellcolor[HTML]{D7FFD7}\textbf{1.69}}\\

\midrule
{\ \ \ $Avg.$}&{\ \ \ 23.36}&{\ \ \ 4.76}&{\ \ \ 2.16} &{\ \ \ 7.98} &{\ \ \ \ 4.99} &{\ \ \ 2.64} &{\ \ \ \cellcolor[HTML]{D7FFD7}\textbf{1.74}} \\

\bottomrule
\end{tabular}
\end{center}
\label{TableCOPD}
\vspace{-16pt}
\end{table*}
\vspace{0cm}
\begin{table*}[!t]\small
\renewcommand\arraystretch{1}
\centering
\caption{Numerical comparison of our proposed framework vs. other existing image registration techniques. The numerical values reflect the TRE metric $(mm)$ for the DIRLab 4DCT dataset. The best results are highlighted in green colour.}
\begin{center}
\begin{tabular}{p{0.45in} p{0.55in}p{0.55in}
                          p{0.55in}p{0.55in}
                          p{0.645in}p{0.55in}
                          p{0.5in}}
                          
\toprule
\rowcolor[HTML]{EFEFEF}
\ 4DCT &{\ \ \ $Init.$}&{\ \ FE}&{\ PDD}&{\ \   VM }
  &{LapIRN} &{ \ \ INR}&{\ \ Ours} \\
\midrule
{\ \ \ \ 01} &{\ \ \ 3.89}  &{\ \ \ 2.20} &{\ \ \ 0.90} &{\ \ \ 1.46} &{\ \ \ \ 1.00} &\cellcolor[HTML]{D7FFD7}{\ \ \ \textbf{0.76}} &{\ \ \ 0.78} \\

{\ \ \ \ 02} &{\ \ \ 4.34}  &{\ \ \ 3.89} &{\ \ \ 0.91} &{\ \ \ 1.51} &{\ \ \ \ 1.28} &\cellcolor[HTML]{D7FFD7}{\ \ \ \textbf{0.76}} &{\ \ \ 0.77}  \\

{\ \ \ \ 03} &{\ \ \ 6.94}  &{\ \ \ 2.71} &{\ \ \ 1.06} &{\ \ \ 2.31} &{\ \ \ \ 2.18} &\cellcolor[HTML]{D7FFD7}\ \ \ \textbf{0.94} &\cellcolor[HTML]{D7FFD7}{\ \ \ \textbf{0.94}}  \\

{\ \ \ \ 04} &{\ \ \ 9.83}  &{\ \ \ 2.95} &{\ \ \ 1.66} &{\ \ \ 2.72} &{\ \ \ \ 3.05} &\cellcolor[HTML]{D7FFD7}{\ \ \ \textbf{1.32}} &{\ \ \ 1.36}  \\

{\ \ \ \ 05} &{\ \ \ 7.48}  &{\ \ \ 3.03} &{\ \ \ 1.68} &{\ \ \ 2.69} &{\ \ \ \ 2.36} &{\ \ \ 1.23} &\cellcolor[HTML]{D7FFD7}{\ \ \ \textbf{1.20}}  \\

{\ \ \ \ 06} &{\ \ \ 10.89} &{\ \ \ 3.36} &{\ \ \ 1.86} &{\ \ \ 3.07} &{\ \ \ \ 1.78} &{\ \ \ 1.09} &\cellcolor[HTML]{D7FFD7}{\ \ \ \textbf{1.06}}  \\

{\ \ \ \ 07} &{\ \ \ 11.03} &{\ \ \ 3.10} &{\ \ \ 1.94} &{\ \ \ 3.01} &{\ \ \ \ 2.24} &{\ \ \ 1.12} &\cellcolor[HTML]{D7FFD7}{\ \ \ \textbf{0.97}}  \\

{\ \ \ \ 08} &{\ \ \ 14.99} &{\ \ \ 2.94} &{\ \ \ 1.79} &{\ \ \ 6.22} &{\ \ \ \ 2.24} &{\ \ \ 1.21} &\cellcolor[HTML]{D7FFD7}{\ \ \ \textbf{1.12}}  \\

{\ \ \ \ 09} &{\ \ \ 7.92}  &{\ \ \ 2.86} &{\ \ \ 1.94} &{\ \ \ 2.94} &{\ \ \ \ 2.26} &{\ \ \ 1.22} &\cellcolor[HTML]{D7FFD7}{\ \ \ \textbf{1.08}}  \\

{\ \ \ \ 10} &{\ \ \ 7.30}  &{\ \ \ 2.99} &{\ \ \ 2.03} &{\ \ \ 3.00} &{\ \ \ \ 1.90} &\cellcolor[HTML]{D7FFD7}{\ \ \ \textbf{1.01}} &{\ \ \ 1.04}  \\
\midrule
{\ \ \ $Avg.$}&{\ \ \ 8.46} &{\ \ \ 3.00} &{\ \ \ 1.57} &{\ \ \ 2.89} &{\ \ \ \ 2.03} &{\ \ \ 1.07} &\cellcolor[HTML]{D7FFD7}{\ \ \ \textbf{1.03}}  \\

\bottomrule
\end{tabular}
\end{center}
\label{Table4DCT}
\vspace{-16pt}
\end{table*}

\subsubsection{Evaluation Protocol.}
We follow the standard protocol for evaluating our framework using the target registration error (TRE) as the performance metric.
The TRE is defined as the distance between a set of manually identified corresponding points, typically referred to as landmarks, in the registered image and the corresponding points in the target image. 
In this paper, the 300 corresponding landmarks provided in the dataset are used to calculate the TRE.
The lowest the TRE the better the registration output.

\subsubsection{Implementation Details.} 
Our method was implemented using PyTorch, and utilised an NVIDIA A100 GPU for computation.
During the training stage, each image pair is trained for 6000 epochs for the COPD dataset, and 3000 epochs for the 4DCT dataset. At each epoch, we randomly sample 15000 points, from the masked image, with only lung information for the COPD dataset (10000 points for the 4DCT dataset). We optimise the network using the Adam optimiser with a fixed learning rate of  1×10$^{-5}$. 
The total time required to register one pair of 3D volumes is 1.7 minutes for COPD dataset and 1.1 minutes for the 4DCT dataset. 
The registration is finished in one-shot without any prior affine registration.

\subsection{Experimental Results}
For evaluating our proposed method, we compare it against five competitive methods for lung CT registration, which is a challenging task due to the superposition of respiratory and cardiac motion. 

\subsubsection{Performance Improvement} 
In Table \ref{TableCOPD}, we report a set of quantitative comparisons in terms of TRE $(mm)$  on the DIRLab COPD dataset between our method and other state-of-the-art methods: 1) FE~\citep{liu2019flownet3d}, 2) PDD~\citep{heinrich2019closing}, 3) VM~\citep{balakrishnan2019voxelmorph}, 4) LapIRN~\citep{mok2020large}, and 5) INR~\citep{wolterink2022implicit}. The COPD dataset has extremely large deformation with an average initial displacement of 23.36 $mm$. The initial TRE of all patients are listed in the second column of the table, denoted as $Init.$.
In a closer look at the results, we observe that our method achieved the lowest TRE in all COPD data pairs, with an average value of 1.74 $mm$. We outperform the second-best method INR by 0.42 $mm$. While other methods have much higher average TRE values ranging from 2.16 $mm$ to 7.98 $mm$. Similarly, in Table \ref{Table4DCT}, we report TRE results on the 4DCT dataset, our method also achieved the lowest average TRE value compared to other methods, and achieved the lowest TRE value on 6 image pairs.

\begin{figure*}[!ht]
\centering
\setlength{\abovecaptionskip}{0pt}
\setlength{\belowcaptionskip}{-2pt}
\subfigure[] {\ \ \includegraphics [width=0.47\linewidth]{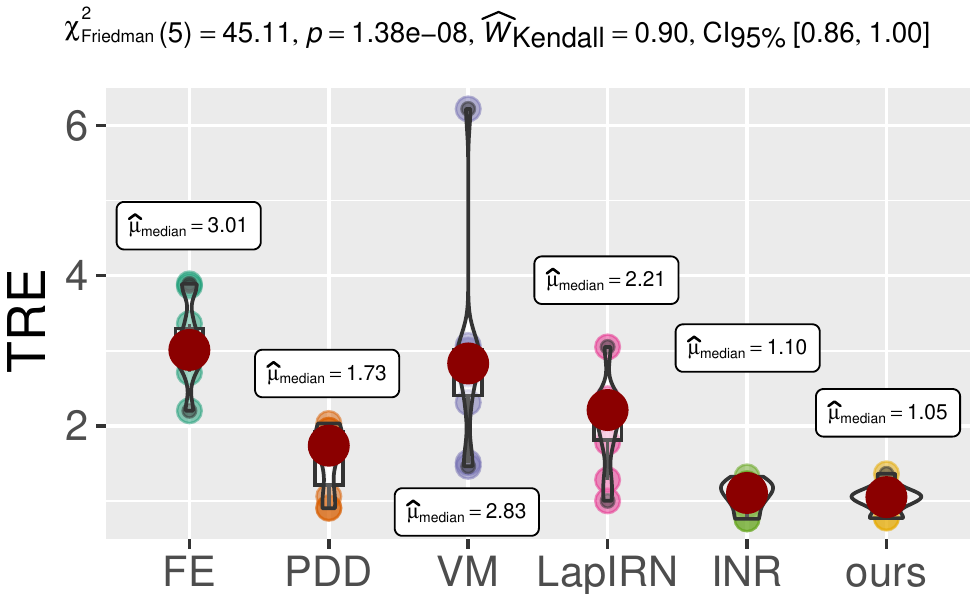}} \ 
\subfigure[] {\includegraphics [width=0.47\linewidth]{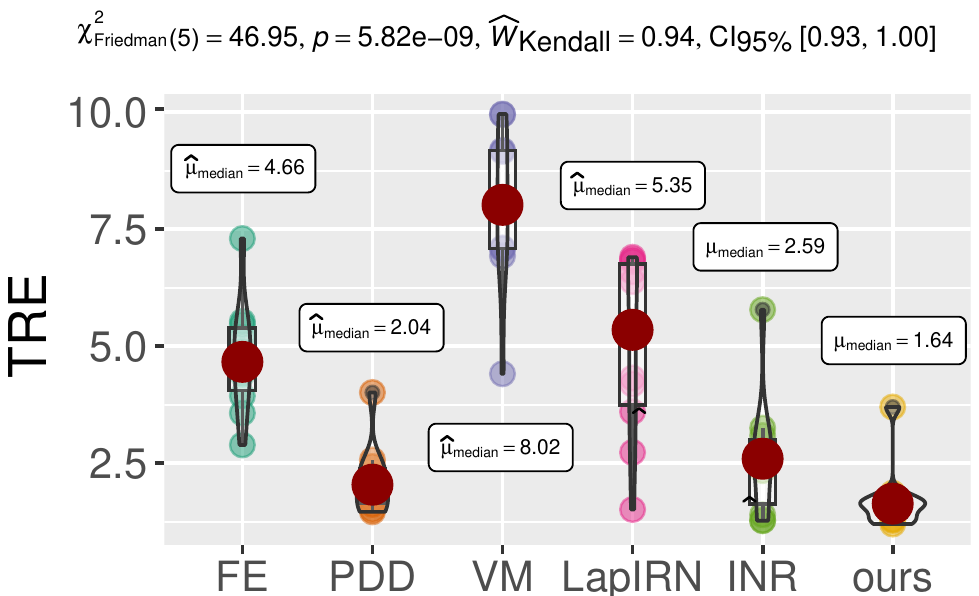}}
\caption{Statistical Analysis of our technique and existing methods. We performed the Friedman test for multiple comparisons along the Wilcoxon test for pair-wise comparison for (a) 4DCT and (b) COPD datasets. }
\label{FigTest}
\vspace{0pt}
\end{figure*}

We also ran a non-parametric test for multiple comparisons using the Friedman test along with Kendall's coefficient of concordance with 95\% confidence intervals as a measure of the effect size for the Friedman test. The statistical analyses for both datasets are reported in Fig.~\ref{FigTest}. We can conclude that there is a statistically significant difference in performance $\chi^2_{Friedman} (5)= 45.11 \& 46.95$ with $p=1.38e\text{-}08$ \& $5.82e\text{-}09$. The effect size is $W_{Kendall} = 0.90 \& 0.94$ with 95\% CI. We then performed  pair-wise comparisons using the non-parametric Wilcoxon test yielding to our technique being statistically significantly different in performance across all compared techniques.

\subsubsection{Qualitative Evaluation.}
We computed the Jacobian determinant and visually represented the results using a color map, where red indicates a positive determinant and blue indicates a negative determinant.  Fig.~\ref{FigJdet} displays the results of selected patients with: No regularisation, Hyperelastic regularisation~\citep{wolterink2022implicit,burger2013hyperelastic}, and our conformal invariant regulariser. We observe that without regularisation, one obtains 
a large number of negative values for the Jacobian determinant. This indicates that the deformation fields obtained were implausible and resulted in topological distortions.  
The Hyperelastic regulariser also reports substantial negative values for the Jacobian determinant (blue areas) along with large expansions. In contrast, our regulariser, reported all positive values yielding to clinical meaningful transformations.
These findings underscore the importance of appropriate regularisation in image registration and further support the superiority of our proposed registration framework.

\begin{figure*}[!h]
\centering
\setlength{\abovecaptionskip}{0pt}
\setlength{\belowcaptionskip}{0pt}
\ \includegraphics [width = 1\linewidth]{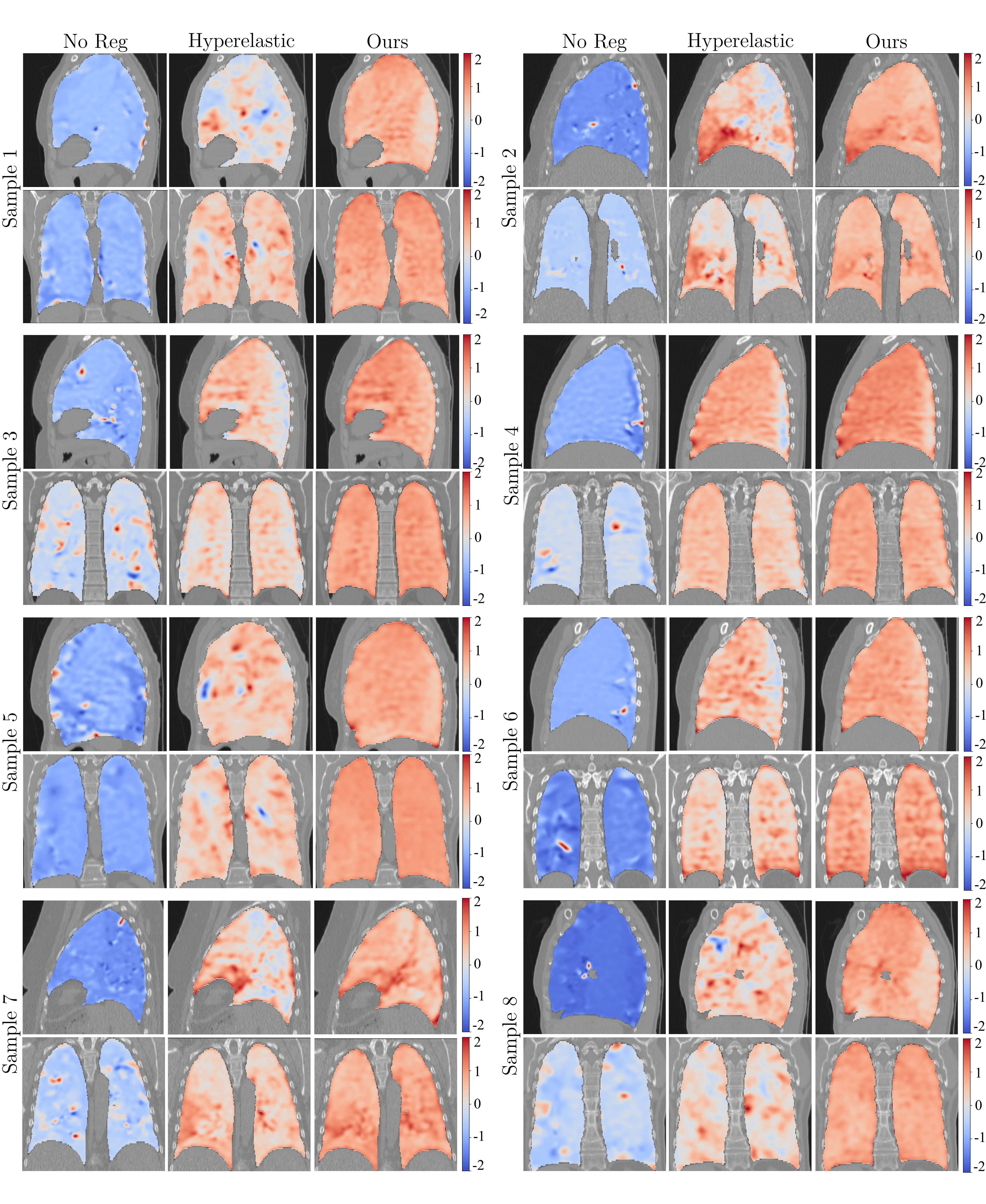}
\caption{Visual comparison in terms of Jacobian determinant of our proposed regulariser vs. no regularisation (`No Reg') and the most widely used hyperelastic regulariser.  }
\label{FigJdet}
\vspace{-2pt}
\end{figure*}
\subsubsection{Results on Different MLPs.}
In this section, we evaluate our conformal invariant regularization by integrating it with different coordinate MLP architectures, specifically MLP with periodic activation function ~\citep{wolterink2022implicit} and Fourier feature mapping  
MLP~\citep{tancik2020fourier}. 
The results of both the COPD dataset and the 4DCT dataset are presented in Table \ref{TableMLP} and the best results are illustrated in bold font. Results show that MLP with periodic activation achieved the best results. 
This experiment provides evidence supporting the notion that MLP with a periodic activation function is capable of learning more high-frequency features, thereby capturing finer details in the images.

\begin{table}[!t]\small
\renewcommand\arraystretch{1}
\centering
\caption{Numerical comparison between MLP with Periodic Activation Function (`Periodic') and MLP with Fourier Feature Mapping (`Fourier'). The numerical values reflect the TRE metric $(mm)$ for both the DIRLab COPD and the 4DCT datasets. Bold font represents better results.}
\begin{center}
\begin{tabular}{p{0.3in} p{0.5in}p{0.5in}
                          p{0.001in}
                          p{0.5in}p{0.5in}}
\toprule
\ &\multicolumn{2}{c}{COPD} &
  &\multicolumn{2}{c}{4DCT}\\
\cmidrule{2-3}
\cmidrule{5-6}
\ &{\ Periodic}&{\ Fourier} &{}
  &{\ Periodic}&{\ Fourier} \\
\midrule

{\ \ \ \ 01} &{\ \ \ \ \ \textbf{1.66}}  &{\ \ \ \ 2.11} & &{\ \ \ \ \ \textbf{0.78}} &{\ \ \ \ 0.81}\\

{\ \ \ \ 02} &{\ \ \ \ \ 3.70}  &{\ \ \ \ \textbf{2.46}} & &{\ \ \ \ \ \textbf{0.77}} &{\ \ \ \ 0.78}\\

{\ \ \ \ 03} &{\ \ \ \ \ \textbf{1.28}}  &{\ \ \ \ 1.56} & &{\ \ \ \ \ \textbf{0.94}} &{\ \ \ \ 1.01}\\ 

{\ \ \ \ 04} &{\ \ \ \ \ \textbf{1.62}}  &{\ \ \ \ 1.88} & &{\ \ \ \ \ 1.36} &{\ \ \ \ \textbf{1.35}} \\

{\ \ \ \ 05} &{\ \ \ \ \ \textbf{1.47}}  &{\ \ \ \ 1.87} & &{\ \ \ \ \ \textbf{1.20}} &{\ \ \ \ 1.23}\\

{\ \ \ \ 06} &{\ \ \ \ \ \textbf{1.86}}  &{\ \ \ \ 2.52} & &{\ \ \ \ \ 1.06} &{\ \ \ \ \textbf{1.03}}\\

{\ \ \ \ 07} &{\ \ \ \ \ \textbf{1.20}}  &{\ \ \ \ 1.61} & &{\ \ \ \ \ \textbf{0.97}} &{\ \ \ \ 1.04}\\

{\ \ \ \ 08} &{\ \ \ \ \ \textbf{1.65}}  &{\ \ \ \ 2.16} & &{\ \ \ \ \ \textbf{1.12}} &{\ \ \ \ 1.23} \\

{\ \ \ \ 09} &{\ \ \ \ \ \textbf{1.30}}  &{\ \ \ \ 1.72} & &{\ \ \ \ \ \textbf{1.08}} &{\ \ \ \ \textbf{1.08}} \\

{\ \ \ \ 10} &{\ \ \ \ \ \textbf{1.69}}  &{\ \ \ \ 2.35} & &{\ \ \ \ \ \textbf{1.04}} &{\ \ \ \ 1.09} \\

\midrule
{\ \ \ $Avg.$}&{\ \ \ \ \ \textbf{1.74}}&{\ \ \ \ 2.02}& &{\ \ \ \ \ \textbf{1.03}} &{\ \ \ \ 1.07}  \\
\bottomrule
\end{tabular}
\end{center}
\label{TableMLP}
\vspace{-16pt}
\end{table}

\section{Conclusion and Discussion}
It is rather difficult to collect large and representative medical image datasets for training the registration networks.
Moreover, existing learning-based registration methods 
lack theoretical guarantees for ensuring the homeomorphism of the deformation field, a crucial requirement in the medical domain. 
Therefore, in this paper, we proposed a new homeomorphic registration framework to solve these problems. We explicitly introduced a novel conformal-invariant hyperelastic regularisation, which yields to a clinically meaningful registration. 
Combining the advantage of neural fields--coordinate-based parameterisation of the physical properties across space and time, we build our framework based on a coordinate MLP with a periodic activation function for learning the high-frequency information of the input image pairs. 
Our extensive experiments on two public lung CT datasets demonstrated the benefits of our registration framework. 


One limitation of the proposed work is the computation time required for the registration process. 
Compared to standard learning-based methods that exhibit rapid inference within a few seconds, our registration framework consumes slightly more time (the total time required to register one pair of 3D volumes is 1.7 minutes for the COPD dataset and 1.1 minutes for the 4DCT dataset). However, it is important to note that our framework achieves one-shot registration avoiding the training with a large dataset and with no need for pre-affine registration. In this perspective, our registration framework is more computationally efficient.
As a consequence, together with the property of homeomorphism, the proposed framework has great potential to be applied in various practical clinical applications, such as tumor tracking in image-guided navigation systems and motion compensation in radiotherapy.

In the future, we plan to investigate mainly two works. Firstly, 
we will conduct a large cohort evaluation on different organs and modalities to assess the generalizability and robustness of our proposed framework. Secondly, we will explore our regulariser by integrating it with other intensity-based registration networks, we aim to enhance the registration accuracy and robustness even further.

In conclusion, we proposed a novel conformal-invariant hyperelastic regulariser for providing a theoretical guarantee of homeomorphism of the deformation field. To the best of our knowledge, we are the first to achieve homeomorphic registration by combining novel conformal-invariant hyperelastic regulariser with coordinate MLP. In future work, we will conduct more experiments and explore the regulariser in conjunction with other intensity-based networks.

\section*{Acknowledgements}
Jing Zou and Jing Qin gratefully acknowledge a General Research Fund of Hong Kong Research Grants Council (15205919).
Lihao Liu acknowledges the financial support from the GSK Ph.D. Scholarship and the Girton College Graduate Research Award at the University of Cambridge.
Angelica I Aviles-Rivero acknowledges the support from the Centre for Mathematical Imaging in Healthcare (CMIH) and the Cantab Capital Institute for the Mathematics of Information (CCIMI) at the University of Cambridge.
Carola-Bibiane Schönlieb acknowledges the support from the Philip Leverhulme Prize, the Royal Society Wolfson Fellowship, the EPSRC Advanced Career Fellowship EP/V029428/1, EPSRC grants EP/S026045/1, EP/T003553/1, EP/N014588/1, and EP/T017961/1, the Wellcome Innovator Awards 215733/Z/19/Z and 221633/Z/20/Z, the European Union Horizon 2020 research and innovation programme under the Marie Skodowska-Curie grant agreement No. 777826 NoMADS, as well as the Cantab Capital Institute for the Mathematics of Information (CCIMI) and the Alan Turing Institute.
%

\bibliographystyle{model2-names.bst}\biboptions{authoryear}
\bibliography{refs}

\end{document}